\documentclass[12pt]{amsart}

\usepackage{amsfonts}
\usepackage{amsmath}
\usepackage{amscd}
\usepackage{amssymb}
\usepackage{amsthm}
\usepackage{bbm}
\usepackage{bm}
\usepackage{enumerate}
\usepackage{hyperref}
\usepackage{latexsym}
\usepackage{mathabx}
\usepackage{mathdots}
\usepackage{mathtools}
\usepackage[left=3cm,top=2cm,right=3cm,bottom = 2cm]{geometry}

\hypersetup{
	colorlinks   = true, 
	urlcolor     = blue, 
	linkcolor    = blue, 
	citecolor   = red 
}

\usepackage{float}
\restylefloat{table}

\let\svthefootnote\thefootnote
\newcommand\freefootnote[1]{%
	\let\thefootnote\relax%
	\footnotetext{#1}%
	\let\thefootnote\svthefootnote%
}

\newtheorem{theorem}{Theorem}[section]
\newtheorem{proposition}[theorem]{Proposition}

\newtheorem{remark}[theorem]{Remark}

\newcommand{\half}{\frac{1}{2}}

\newcommand{\vect}[2]{
	\begin{pmatrix}
		#1 
		\\
		#2
	\end{pmatrix}
}

\newcommand{\mat}[4]{
	\begin{pmatrix}
		#1 & #2 
		\\
		#3 &#4
	\end{pmatrix}
}

\newcommand{\svect}[2]{
\left(\begin{smallmatrix}#1 \\#2 \end{smallmatrix}\right)
}

\newcommand{\smat}[4]{
\left(\begin{smallmatrix}
		#1 & #2 
		\\
		#3 &#4 
\end{smallmatrix}\right)
}

\DeclareMathOperator{\sgn}{sgn}

\begin{document}

\title{Inertial Transformations and the Nonexistence of Tachyons for Spacetime Dimension Greater than Two}
\author{David Acton}
\email{davidacton95@gmail.com}
\author{Owen Doyle}
\email{owend745@gmail.com}
\author{Michael P. Tuite}
\email{michael.tuite@universityofgalway.ie } 
\address{School of Mathematical and Statistical Sciences\\ 
University of Galway, Galway H91 TK33, Ireland}

\date{}

\begin{abstract}
We consider real linear transformations between two inertial frames with constant relative speed $v$ in a $d$-dimensional spacetime where light moves with constant speed $c=1$ (for some chosen units) in all frames. For $d=2$ we show that the standard relative velocity formula holds and that any associated anisotropic conformal factor is multiplicative under composition of  inertial transformations for $|v|\neq 1$. Assuming that the inertial transformation matrix is continuous in a neighbourhood of $v=0$  and differentiable at $v=0$, we determine the conformal factor for all $|v|\neq 1$. For an isotropic spacetime, the general solution reduces to the standard $d=2$ Lorentz transformation for $|v|<1$ or to a Tachyonic transformation for $|v|>1$, first described by Parker in 1969. For $d>2$ we show that no Tachyonic-like inertial transformations exist which are compatible with constant light speed. 
\end{abstract}

\maketitle

\section{Introduction}
There is a long but sporadic research history concerning the extension of Einstein's special relativity \cite{Einstein.1905} to superluminal or tachyonic particles. Early speculations of Sommerfeld \cite{Sommerfeld} in 1904 were followed much later by Bilaniuk et al \cite{bilaniuk_deshpande_sudarshan_1962} in 1962, Feinberg \cite{Feinberg1967} in 1967, Parker\footnote{Most of our results were completed before we discovered Parker's paper \cite{Parker1969} which appears to have been overlooked by many recent authors.} \cite{Parker1969} in 1969 and many others since e.g. \cite{Goldoni1972, Lord_Shankara_1977, Marchildon_Antippa_Everett_1983, Sutherland_Shepanski_1986, ParkPark1996, hill_cox_2012,  Andréka_Madarász_Németi_Székely_2013, Jin_Lazar_2014,Dragan_Ekert_2020}. 
Most of this work was based on modifications of Lorentz transformations by either (i) allowing for an imaginary Lorentz factor $\gamma(v)=(1-v^{2})^{-\half}$ for $|v|>1$  (where $c=1$) leading to imaginary spacetime coordinates and mass or (ii)  replacing $\gamma(v)$ by $\pm (v^{2}-1)^{-\half}$. Here we consider the problem from first principles based on Einstein's axioms alone. 

In Section~2 we consider all real linear Inertial Transformations (ITs) between two inertial frames in $d=2$ dimensional spacetime with constant relative velocity $v$. We initially consider two standard relativistic axioms (I): Free particles move with constant (frame dependent) velocity in every inertial frame and (II):
The speed of light $c=1$ is constant in every inertial frame. This is similar to Parker's approach \cite{Parker1969} except that we allow for spatial anisotropy expressed in terms of an overall conformal factor $\phi(v)$. We show that Lorentz-like ITs with $|v|<1$ and Tachyonic-like ITs with  $|v|>1$ exist. Composing ITs we find that the standard addition of velocity formula holds and that $\phi(v)$ is multiplicative under composition for all $|v|\neq 1$. In Section~3 we discuss a "flipping" symmetry relating $d=2$ ITs for $v$ and $v^{-1}$. This together with composition multiplicativity of $\phi(v)$
allows us to determine its general form for $|v|\neq 1$ for all $d=2$ ITs assuming that $\phi(v)$ is continuous in a neighbourhood of $v=0$ and differentiable at $v=0$. We then consider Axiom~(III): Spacetime is spatially isotropic, which implies trivial unit conformal factors. We recover the standard Lorentzian IT when $|v|<1$ and the Parker isotropic Tachyonic IT  \cite{Parker1969} when $|v|>1$ (which also appeared later in \cite{Marchildon_Antippa_Everett_1983, Jin_Lazar_2014}). Section~4 is a brief discussion on the energy-momentum in this setting. Section~5 shows that for $d>2$, Axioms~I and II are incompatible for $|v|>1$. 
The argument is elementary and follows from an analysis of light which in one frame is moving perpendicularly to the direction of relative motion to another frame.

\section{Inertial Transformations in Dimension $d=2$ Spacetime  }
Let us consider the general form of an Inertial Transformation (IT) between  inertial frames in spacetime dimension $d=2$, subject to the standard axioms of special relativity \cite{Einstein.1905}. Let $S$ and $S'$ denote two inertial reference frames where $S'$ moves at a constant velocity $v$ relative to $S$ and $S$ moves at a constant velocity $-v$ relative to $S'$. Every spacetime event is described by real coordinates $(x,t)$ in $S$  and real coordinates $(x',t')$ in $S'$ where the coordinate axes are chosen such that $(x,t)=(0,0)$ and $(x',t')=(0,0)$ describe the same event. We also assume that $S'$ is identified with $S$ for $v=0$. We consider the following axioms: 
\begin{enumerate}[I.]
\item Free particles move with constant (frame dependent) velocity in every inertial frame. 
\item The speed of light is $c=1$ (for a choice of units) in every inertial frame. 
\item  Spacetime is spatially isotropic.
\end{enumerate}
We initially examine the consequences of Axioms~I and II. These are sometimes  the only axioms cited in elementary expositions of special relativity although Einstein did exploit spatial isotropy without formally stating it as an initial axiom \cite{Einstein.1905}.

Axiom~I implies that a particle's motion is described by straight lines in $S$ and $S'$. Thus $(x,t)$ and $(x',t')$ are related by a linear IT
\begin{align}
\label{eq:xMap}
\vect{x}{t}=\mathbf{G}(v)\vect{x'}{t'} \mbox{ for } \mathbf{G}(v):=\mat{A(v)}{B(v)}{C(v)}{D(v)},
\end{align}
where $A,B,C,D$ are real functions of $v$ and $\mathbf{G}(0)=\mathbf{I}$, the identity matrix.  
Axiom~I implies that a particle at rest in $S'$ with coordinates
$(0,t')$ has coordinates $(vt,t)$ in $S$ implying $B(v)=vD(v)$. 
Similarly, a particle at rest in $S$ with coordinates
$(0,t)$ has coordinates $(-vt',t')$ in $S'$ implying $B(v)=vA(v)$. 
Axiom~II implies that a photon  with coordinates $(t,t)$ in $S$   has coordinates $(t',t'$) in $S'$ so that $B(v)=C(v)$. 
Therefore
\begin{align}
\label{eq:MA}
\mathbf{G}(v)=A(v)\mat{1}{v}{v}{1},
\end{align} 
where  $A(0)=1$. 
Axiom~I implies that $\svect{x'}{t'}=\mathbf{G}(-v)\svect{x}{t}$ with
$\mathbf{G}(v)\mathbf{G}(-v)=\mathbf{I}$. Hence 
\begin{align}
\label{eq:AvAminv}
A(v)A(-v)(1-v^{2})=1.
\end{align}
\begin{remark}\label{rem:AG_deform}
Notice that $A(v)\neq 0$ for $|v|\neq 1$ and $A(v)$ is singular for at least one value $v\in\{-1,1\}$. In particular, we therefore cannot continuously deform $\mathbf{G}(v)$ from $v=0$ to $v=\pm \infty$.
\end{remark}	
\begin{remark}
In many elementary derivations of  $d=2$ Lorentz transformations (e.g. \cite{Rindler, Taylor, Williams})  it is often an unstated assumption that $A(v)$ is an even function in $v$ so that  \eqref{eq:AvAminv} implies that $|v|<1$ and 
\begin{align}
	\label{eq:gammav}
	A(v)=\gamma(v):=\left(1-v^{2}\right)^{-\half},
\end{align}
the standard gamma factor  (in units with $c=1$) with  Lorentz transformation \eqref{eq:MA}.
Later we will find the general solution to \eqref{eq:AvAminv} assuming only that $\mathbf{G}(v)$  is continuous in a neighbourhood of $v=0$ and differentiable at $v=0$. 
\end{remark}
Define, for $|v|\neq 1$, a real conformal factor  $\phi(v)$:
\begin{align}
\label{eq:phidef}
\phi(v):=\begin{cases*}
	A(v)\left(1-v^{2}\right)^{\half} & for $|v|<1$,
	\\
	vA(v)\left(1-v^{-2}\right)^{\half} & for  $|v|>1$,
\end{cases*}
\end{align}
where $\phi(v)\phi(-v)=1$ for $|v|\neq 1$ and $\phi(0)=1$ from \eqref{eq:AvAminv}. 
Thus the $d=2$ ITs satisfying Axioms I and II are described  by
\begin{alignat}{3}
	\label{eq:xtL}
	&x = \frac{\phi(v)}{\sqrt{1-v^{2}}}(x'+vt'), \quad
	&&t  = \frac{\phi(v)}{\sqrt{1-v^{2}}}(t'+vx'),\quad
	&&\mbox{for }|v|<1,
	\\
	\label{eq:xtT}
	&x = \frac{\phi(v)}{\sqrt{1-v^{-2}}}(t'+v^{-1}x'),\quad 
	&&t = \frac{\phi(v)}{\sqrt{1-v^{-2}}}(x'+v^{-1}t'),\quad 
	&&\mbox{for }|v|>1.
\end{alignat}
It is useful to define the following matrices for real $\sigma$: 
\begin{align*}
\mathbf{L}(\sigma):=\mat{\cosh\sigma}{\sinh\sigma}{\sinh\sigma}{\cosh\sigma},
\quad
\mathbf{T}(\sigma):=\mathbf{F}\mathbf{L}(\sigma)
=\mathbf{L}(\sigma)\mathbf{F}=\mat{\sinh\sigma}{\cosh\sigma}
{\cosh\sigma}{\sinh\sigma},
\end{align*}
where  $\mathbf{F}=\smat{0}{1}{1}{0}$ is the ``flipping" matrix. 
Note that
$\det \mathbf{L}(\sigma) =1 $ and $\det \mathbf{T}(\sigma)=-1$. Furthermore, all $\mathbf{L}$ and $\mathbf{T}$ matrices commute
and obey the relations \cite{Parker1969}
\begin{align}\label{eq:LTmult}
\mathbf{L}(\sigma_{1})\mathbf{L}(\sigma_{2})=
\mathbf{T}(\sigma_{1})\mathbf{T}(\sigma_{2})=
\mathbf{L}(\sigma_{1}+\sigma_{2}),
\quad 
\mathbf{L}(\sigma_{1})\mathbf{T}(\sigma_{2})=\mathbf{T}(\sigma_{1}+\sigma_{2}).
\end{align}
From \eqref{eq:MA} and \eqref{eq:phidef} we may therefore write every $d=2$ IT matrix as follows:
\begin{align}
\label{eq:MLT}
\mathbf{G}(v)=
\begin{cases*}
	\phi(v)\mathbf{L}(\psi)& where $\psi := \tanh^{-1}(v)$ for $|v|<1$,
	\\
	\phi(v)\mathbf{T}(\chi)& where $\chi := \tanh^{-1}(v^{-1})$ for $|v|>1$.
\end{cases*}
\end{align}
If $\phi(v)=1$ then $\mathbf{L}(\psi)$ is the standard Lorentz transformation describing inertial frame transformations when $|v|<1$.
With $|v|>1$ we refer to $\mathbf{T}(\chi)$ as a Tachyonic-like transformation. Since $\mathbf{F}\mathbf{G}(v)=\mathbf{G}(v)\mathbf{F}$ we find $\smat{x}{t}{t}{x}=\mathbf{G}(v)\smat{x'}{t'}{t'}{x'}$ whose determinant leads to a generalised Minkowski invariance relation
	\begin{align}
		\label{eq:Mink}
		x^{2}-t^{2}  = \det \mathbf{G}(v)\, \left(x'^{2}-t'^{2}\right).	\end{align}
Note that $\det\mathbf{G}(v)=\phi(v)^2>0$ for $|v|<1$  and $\det\mathbf{G}(v)=-\phi(v)^2<0$ for $|v|>1$. Thus the sign of $\det\mathbf{G}(v)$ indicates whether $|v|<1$ or $|v|>1$.

We now consider the consequences of  composing two ITs arising from coordinate changes between three inertial frames $S,S'$ and $S''$. Let $S'$ have velocity $v_{1}$ relative to $S$ and let $S''$ have velocity $v_{2}$ relative to $S'$ so that
\begin{align*}
\vect{x}{t}=\mathbf{G}(v_{1})\vect{x'}{t'}=\mathbf{G}(v_{1})\mathbf{G}(v_{2})\vect{x''}{t''}.
\end{align*}
Let $v_{3}$ denote the velocity of $S''$ relative to $S$ so that $\mathbf{G}(v_{3})=\mathbf{G}(v_{1})\mathbf{G}(v_{2})$.
\begin{proposition}\label{prop:relvel}
 $\mathbf{G}(v_{3})=\mathbf{G}(v_{1})\mathbf{G}(v_{2})=\mathbf{G}(v_{2})\mathbf{G}(v_{1})$ for all $|v_{1}|,|v_{2}| \neq 1$ where $v_{3}$ is given by the general relative velocity formula
\begin{align}
\label{eq:relvel}	
v_{3}=\frac{v_{1}+v_{2}}{1+v_{1}v_{2}}, 
\end{align}
with a multiplicative conformal factor:
\begin{align}\label{eq:relphi}
\phi(v_{3})=\phi(v_{1})\phi(v_{2}).
\end{align}
\end{proposition}
\begin{proof}
The relations \eqref{eq:LTmult} imply that $\mathbf{G}(v_{1})\mathbf{G}(v_{2})=\mathbf{G}(v_{2})\mathbf{G}(v_{1})$ for all $v_{1},v_{2}$. By relabelling, we may therefore  assume that $|v_{1}|\le |v_{2}|$.
There are thus three cases to consider: (i) $|v_{1}|,|v_{2}| <1$, (ii) $|v_{1}|<1<|v_{2}| $  and (iii)  $1<|v_{1}|,|v_{2}|$. 
In case (i) we have
\begin{align*}
	\mathbf{G}(v_{3}) = \phi(v_{1})\phi(v_{2})
	\mathbf{L}(\psi_{1}) \mathbf{L}(\psi_{2})=
	\phi(v_{3})\mathbf{L}(\psi_{3}),
\end{align*}
with 
$\psi_{1}=\tanh^{-1}(v_{1})$, $\psi_{2}=\tanh^{-1}(v_{2})$ and
$\psi_{3}=\psi_{1}+\psi_{2}$ from \eqref{eq:LTmult} so that
$\tanh \psi_{3}=\frac{v_{1}+v_{2}}{1+v_{1}v_{2}}$  (as in standard special relativity) 
and $\phi(v_{3})=\phi(v_{1})\phi(v_{2})$.

For case (ii)  we have $\det \mathbf{G}(v_{3})=\det \mathbf{G}(v_{1}) \det \mathbf{G}(v_{2})<0$ so that $|v_{3}|>1$. Thus
\begin{align*}
	\mathbf{G}(v_{3}) = \phi(v_{1})\phi(v_{2})
	\mathbf{L}(\psi_{1}) \mathbf{T}(\chi_{2})=
\phi(v_{3})\mathbf{T}(\chi_{3}),
\end{align*}
with $\psi_{1}=\tanh^{-1}(v_{1})$, $\chi_{2}=\tanh^{-1}(v_{2}^{-1})$,  $\phi(v_{3})=\phi(v_{1})\phi(v_{2})$ and  $\chi_{3}=\psi_{1}+\chi_{2}$ from \eqref{eq:LTmult}. Hence \eqref{eq:relvel} holds since
\begin{align*}
\frac{1}{v_{3}}=	\tanh \chi_{3}=&
\frac{\tanh\psi_{1}+\tanh\chi_{2}}
{1+\tanh\psi_{1}\tanh\chi_{2}}
=\frac{v_{1}+v_{2}^{-1}}{1+v_{1}v_{2}^{-1}}
=\frac{1+v_{1}v_{2}}{v_{1}+v_{2}}.
\end{align*}
For case (iii) we have $\det \mathbf{G}(v_{3})>0$ implying that $|v_{3}|<1$. Then we find 
\begin{align*}
\mathbf{G}(v_{3}) = \phi(v_{1})\phi(v_{2})
\mathbf{T}(\chi_{1}) \mathbf{T}(\chi_{2})=
\phi(v_{3})\mathbf{L}(\psi_{3}),
\end{align*}
with $\phi(v_{3})=\phi(v_{1})\phi(v_{2})$ and $\psi_{3}=\chi_{1}+\chi_{2}$ from \eqref{eq:LTmult}. Hence 
\begin{align*}
v_{3}=\tanh \psi_{3}=&
\frac{\tanh\chi_{1}+\tanh\chi_{2}}
{1+\tanh\chi_{1}\tanh\chi_{2}}
=\frac{v_{1}^{-1}+v_{2}^{-1}}{1+v_{1}^{-1}v_{2}^{-1}}
=\frac{v_{1}+v_{2}}{1+v_{1}v_{2}}.
\end{align*}
\end{proof}
\begin{remark}
The relative velocity relation \eqref{eq:relvel} with trivial conformal factor $\phi(v)=1$ for $|v|<1$ is the standard one of special relativity. Its extension to superluminal velocities $|v|>1$ with $\phi(v)=1$ is discussed in \cite{Goldoni1972},  \cite{Lord_Shankara_1977} and \cite{hill_cox_2012} but with Tachyonic transformations differing from \eqref{eq:xtT} in each case cf. Remark~\ref{rem:compare} below.
\end{remark}

\section{Properties of the Conformal Factor}
\begin{proposition}\label{prop:phicon}
Suppose that $\phi(v)$ is continuous in an open neighbourhood of $v=0$. Then $\phi(v)$ is continuous for all $|v|\neq 1$.
\end{proposition}
\begin{proof}
By assumption, $\phi$ is continuous on $\Delta_{0}:=(-\delta_{0},\delta_{0}) $ for some $0<\delta_{0}<1$.
From \eqref{eq:relphi} it follows that 
$\phi\left(\frac{u+v}{1+uv}\right)=\phi(u)\phi(v)$ is continuous in $\frac{u+v}{1+uv}$ for all  $u,v\in \Delta_{0}$. 
Thus $\phi$ is continuous on $\Delta_{1}:=(-\delta_{1},\delta_{1}) $  for $\delta_{1}=f(\delta_{0})$  with $f(x):=\frac{2x}{1+x^{2}}$ where the $\pm \delta_{1}$ interval endpoints correspond to $u=v=\pm \delta_{0}$.
 $f$ defines a monotonically increasing map from $(-1,1)$ to itself and   $\Delta_{1}=f(\Delta_{0})\supset\Delta_{0}$. 
Iterating this process we find that $\phi$ is continuous on $\Delta_{n+1}:=f(\Delta_{n})\supset \Delta_{n}$ for all $n\ge 0$. 
Since the fixed points of $ f(x)$ are $0,\pm 1$ we conclude that $\phi(v)$ is continuous for all $|v|<1$.
Eqn. \eqref{eq:relphi} further implies that for $|v|>1$ then $\phi\left(\frac{u+v}{1+uv}\right)=\phi(u)\phi(v)$ is continuous in $u$ for all $|u|<1$. Since 
$\left| \frac{u+v}{1+uv}\right |>1$ we find $\phi(v)$ is  also continuous for all $|v|>1$.
\end{proof}
Using Proposition~\ref{prop:phicon} we compute  $\lim_{u\rightarrow \pm\infty}\phi(u)\phi(v)=\lim_{u\rightarrow \pm\infty}\phi\left(\frac{u+v}{1+uv}\right)$ to find
\begin{align}\label{eq:phiu_uinv}
\phi(\pm\infty)\phi(v)=\phi(v^{-1}),
\end{align}
for all $|v|\neq 1$. Then $\phi(0)=1$  implies 
$\phi(+\infty)=\phi(-\infty)\in\{\pm 1\}$. Note that the sign of $\phi(\infty)$ is not determined from the condition $\phi(0)=1$ cf. Remark~\ref{rem:AG_deform}.

 From \eqref{eq:MLT} we find that $\svect{x}{t}
\rightarrow\phi(\infty)\mathbf{F}\svect{x'}{t'}=
\phi(\infty)\svect{t'}{x'}$ as $|v|\rightarrow\infty$. Therefore for $|v|>1$ we make a conventional\footnote{The alternative choice of $\phi(\infty)=-1$ is not physically distinguishable from our choice but rather just represents alternative $S'$ axis directions when $|v|>1$. This is at variance with remarks in \cite{hill_cox_2012,Andréka_Madarász_Németi_Székely_2013}.} 
choice of the direction of the $S'$ axes so that $\phi(\infty)=1$ (in analogy with the convention that $S=S'$ for $v=0$ which implies $\phi(0)=1$). 
With this convention, we find from \eqref{eq:phiu_uinv} that
\begin{align}
\label{eq:phivinv}
\phi(v)=\phi(v^{-1}).
\end{align}
Furthermore, we obtain the following  $d=2$ IT ``flipping" symmetry\footnote{This suggests the term ``Ayoade symmetry".} 
\begin{align}
\label{eq:Flip}
\mathbf{G}(v^{-1})=\mathbf{F}\mathbf{G}(v)=\mathbf{G}(v)\mathbf{F}, \end{align}
i.e. $A(v^{-1})=vA(v)$. In particular, for $|v|=\infty$ we find  $\svect{x}{t}=\svect{t'}{x'}$ from \eqref{eq:xtT} so that the space and time variables are flipped \cite{Lord_Shankara_1977, Andréka_Madarász_Németi_Székely_2013} irrespective of the sign of $v=\pm \infty$. 
Note that no physical observation can distinguish  $v=\infty$ from $v=-\infty$.
With $|v|=\infty$, the worldline of a particle at rest at $x'=0$ in $S'$ for all $t'$ is observed as an instantaneous spacelike worldline $t=0$ in $S$ for all $x$ and similarly, for the observation in $S'$  of a particle at rest at $x=0$ in $S$. 
More generally, suppose a tachyon has velocity $u>1$ relative to a frame $S$. 
Let $S'$ be another frame moving with velocity $v=u^{-1}+\varepsilon$ relative to $S$ where $|v|<1$ for sufficiently small $\varepsilon$.
The tachyon velocity in $S'$ is
\begin{align*}
	u'=\frac{u-v}{1-uv}= -\frac{1}{\varepsilon}\left(1-u^{-2}\right)+u^{-1}\rightarrow \mp \infty,
\end{align*}	
as $\varepsilon\rightarrow 0\pm$. Thus, with $v=u^{-1}$, the particle has infinite velocity in $S'$ where  $u'=\infty$ and $u'=-\infty$ are physically indistinguishable scenarios.

We now determine the general form of $\phi(v)$ on further assuming that $\phi(v)$ is differentiable at $v=0$ and with $\phi(\infty)=1$ by the above convention.
\begin{proposition}\label{prop:phi}
Assume that $\phi'(0)$ exists. Then $\phi'(v)$ exists for all $|v|\neq 1$ and
	\begin{align}
		\label{eq:phiv}
		\phi(v)=\left| \frac{1+v}{1-v}\right|^{\alpha},
	\end{align}
for $|v|\neq 1$ where $\alpha:=\half\phi'(0)$. 
\end{proposition}
\begin{proof}
With $|v|\neq 1$ and arbitrarily small $\varepsilon$ consider
\begin{align*}
\frac{1}{\varepsilon}\left(\phi(v+\varepsilon)-\phi(v)\right)
=& \phi(v)\frac{1}{\varepsilon}\left( \phi(-v)\phi(v+\varepsilon)-1\right)
\\
= & \phi(v)\frac{1}{\varepsilon}\left( \phi\left(\frac{\varepsilon}{1-v(v+\varepsilon) }\right)-1\right),
\end{align*}
using \eqref{eq:relphi}. By assumption, the zero limit of $\varepsilon$ exists on the RHS implying that
\begin{align}\label{eq:phip}
	\phi'(v)=\phi(v)\frac{2\alpha}{1-v^{2}},
\end{align}
for all $|v|\neq 1$.
For $|v|<1$ and recalling $\phi(0)=1$ then \eqref{eq:phip} integrates to
\begin{align}
	\label{eq:phiL}
	\phi(v)=\left(\frac{1+v}{1-v}\right)^{\alpha}.
\end{align}
Eqn. \eqref{eq:phivinv} implies that for $|v|>1$
\begin{align*}
	\phi(v)=\left(\frac{1+v^{-1}}{1-v^{-1}}\right)^{\alpha}
	=\left| \frac{1+v}{1-v}\right|^{\alpha}.
\end{align*}
\end{proof}
\begin{remark}
From \eqref{eq:phidef} and \eqref{eq:Flip} we have
	\begin{align*}
		A(v):=\begin{cases*}
			\left(1+v\right)^{\alpha-\half}\left(1-v\right)^{-\alpha-\half}  & for $|v|<1$,
			\\
			v^{-1}\left(1+v^{-1}\right)^{\alpha-\half} \left(1-v^{-1}\right)^{-\alpha-\half}  & for  $|v|>1$.
		\end{cases*}
	\end{align*}
	Note that $A(v)$ is singular at $v=1$ for $\alpha>-\half$ and at $v=-1$ for $\alpha<\half$ confirming $\mathbf{G}(v)$ is singular for at least one value $v\in\{-1,1\}$ as in Remark~\ref{rem:AG_deform}.
\end{remark}
 We now consider Axiom~III. Consider a clock at rest at the origin of $S'$ moving at velocity $v$ for $|v|<1$ with respect to $S$. Let $\tau$ be the proper time interval between two ticks. This  is measured in $S$  with dilated time interval 
\begin{align}\label{eq:Ttau}
	T(v)=(1-v^{2})^{-\half}\phi(v)\tau 
	=(1-v)^{-\half-\alpha}(1+v)^{-\half+\alpha}\tau.
\end{align}
If  $\alpha\neq 0$ then $T(v)\neq T(-v)$ so that the time dilation depends on the direction of motion 
i.e. the spacetime is spatially anisotropic. Thus the observed decay rate of an unstable particle would depend on its direction of motion along the $X$ axis. Therefore we are led to adopt Axiom~III so that $\alpha=0$ with trivial conformal factor $\phi(v)=1$ for all $v$. This leads to the standard 2-d Lorentz transformation $\mathbf{G}(v)=\mathbf{L}(\psi)$ for $\psi=\tanh^{-1}(v)$ for $|v|<1$ from \eqref{eq:MLT} described in \eqref{eq:xtL} with $\phi(v)=1$.
However, Axiom~III still allows for tachyonic transformations $\mathbf{G}(v) = \mathbf{T}(\chi )$ for $\chi=\tanh^{-1}(v^{-1})$ for $|v|>1$ described in \eqref{eq:xtT} with $\phi(v)=1$ given by
\begin{align}\label{eq:xtTIII}
\vect{x}{t} = \frac{1}{\sqrt{1-v^{-2}}}\vect{t'+v^{-1}x'}{x'+v^{-1}t'}
	=\frac{\sgn{v}}{\sqrt{v^{2}-1}}\vect{x'+vt'}{t'+vx'},\quad |v|>1.
\end{align}
Thus  \eqref{eq:AvAminv} has the standard isotropic even Lorentz solution $A(v)=\gamma(v)$ of \eqref{eq:gammav} for $|v|<1$ but also an isotropic odd tachyonic solution  for $|v|>1$ given by\footnote{The other odd solution $A(v)=-\sgn{v}/\sqrt{v^{2}-1}$ corresponds to the alternative choice of $S'$ axes directions for $\phi(\infty)=-1$ where $(x',t')=-(t,x)$ for $|v|=\infty$.}
\begin{align*}
	A(v)=\frac{1}{v\sqrt{1-v^{-2}}}=\frac{\sgn{v}}{\sqrt{v^{2}-1}}.
\end{align*}
\begin{remark}\label{rem:compare}
We note that \eqref{eq:xtTIII} agrees with tachyonic ITs apparently first described by Parker \cite{Parker1969} in 1969 and later again in \cite{Marchildon_Antippa_Everett_1983, Jin_Lazar_2014,Dragan_Ekert_2020}
but not those described by Goldoni \cite{Goldoni1972} in 1972  and later in 
\cite{Lord_Shankara_1977, hill_cox_2012,Andréka_Madarász_Németi_Székely_2013} which do not include the necessary $\sgn{v}$ factors required for consistency in the  composition of  ITs as in Proposition~\ref{prop:relvel}. 
\end{remark}

	\section{$d=2$ Energy-Momentum}
Let us assume Axioms~I--III.   Suppose that a particle has velocity $u$ relative to  $S$ and coordinates $(x,t)$. Let $S_{0}$ be the particle's rest frame with coordinates $(x_{0},t_{0})=(0,\tau)$ for proper time $\tau$ and let $E_{0}=m_{0}$ be the rest mass energy (in units where $c=1$) and $p_{0}=0$ the momentum. We define the energy-momentum vector in $S$ by
\begin{align}\label{eq:pEdef}
	\vect{p}{E}
	= m_0 \frac{d}{d \tau}\vect{x}{t}=\mathbf{G}(u)\vect{0}{m_{0}},
\end{align}
for $\mathbf{G}(u)$ of \eqref{eq:MLT} with $\phi(u)=1$. 
Similarly, the energy-momentum $(p',E')$ in another frame $S'$ moving at velocity $v$ relative to $S$ is given by $\svect{p}{E}= \mathbf{G}(v)\svect{p'}{E'}$ (so that $(p',E')=(0,m_{0})$ for $v=u$). If $|u|<1$ then \eqref{eq:pEdef} is the usual Einstein energy-momentum $(p,E)=(mu,m)$ for relativistic mass $m=(1-u^{2})^{-\half}m_{0}$. For $|u|>1$ we find, similarly to \eqref{eq:xtT}, that
\begin{align}\label{eq:pET}
	\vect{p}{E}
	=\frac{1}{\sqrt{1-u^{-2}}}\vect{m_{0}}{m_{0}u^{-1}}
	=\frac{\sgn{u}}{\sqrt{u^{2}-1}}\vect{m_{0}u}{m_{0}}.
\end{align}
In this case  the momentum $p$ is always positive whereas $\sgn E=\sgn u$. Furthermore, $(p,E)\rightarrow (m_{0},0)$  as $|u|\rightarrow \infty$ so that the momentum and energy variables are flipped relative to the rest frame values just as for the space time variables. This is consistent with our earlier remarks that  $u= \infty$ and $u=-\infty$ cannot be physically distinguished so that $p$ does not depend on the direction of $u$. 
Finally, we note that a generalised Minkowski invariance relation like  \eqref{eq:Mink}  applies to energy-momentum with
\begin{align}
	E^{2}-p^{2}  = \det \mathbf{G}(v)\, m_{0}^2,
\end{align}
where 
$\det \mathbf{G}(v)=1$ for  $|v| < 1$ and $\det \mathbf{G}(v)=-1$ for $|v| >1$.

\section{Absence of Tachyonic-like ITs for $d>2$}
We now show that Axioms~I and II imply that only Lorentz-like ITs are possible for spacetime dimension $d>2$. 
The argument does not invoke causality or the singular behaviour of ITs for $|v|=1$ but rather we show, in an elementary way, that Axioms~ I and II are incompatible for Tachyonic-like ITs.
Thus tachyons cannot exist in spacetimes with $d>2$ in contradiction to speculations in many papers\footnote{We do note that the failure of Axiom~II is discussed in \cite{Marchildon_Antippa_Everett_1983} for 4 spacetime dimensions.} e.g. \cite{bilaniuk_deshpande_sudarshan_1962,Feinberg1967,Goldoni1972, Lord_Shankara_1977,Sutherland_Shepanski_1986, ParkPark1996}.

It is sufficient to consider the $d=3$ case. Let $S$ and $S'$ be inertial frames with coordinates $(x,y,t)$ and $(x',y',t')$, respectively, where $S'$ is moving at velocity $v$ relative to $S$ along a common $X,X'$ axis direction with $S=S'$ for $v=0$. Following a similar discussion to that of Section~2, we consider the following standard scenarios:
\begin{itemize}
	\item A particle at rest in $S'$ with coordinates $(0,y',t')$  and coordinates $(vt,y,t)$ in $S$ for constant $y,y'$. 
	\item A particle at rest  in $S$ with coordinates $(0,y,t)$ and coordinates $(-vt',y',t')$ in $S'$ for constant $y,y'$.
	\item A photon (with speed $c=1$) moving in the $X$ direction with coordinates $(t,0,t)$  in $S$ and coordinates  $(t',0,t')$ in $S'$.
\end{itemize}
Applying Axioms~I and II as before we find the IT is described by
\begin{align}\label{eq:MAB}
	\begin{pmatrix}
		x \\
		t 
	\end{pmatrix}
	=
	A(v)\begin{pmatrix}
		1 & v \\
		v  & 1
	\end{pmatrix}
	\begin{pmatrix}
		x'\\
		t'
	\end{pmatrix},\quad y=B(v) y',
\end{align}
for real $A(v),B(v)$ where $A(v)$ obeys \eqref{eq:AvAminv}  and where $B(0)=1$ and $B(v)B(-v)=1$ for $|v|\neq 1$. Note that $A(v),B(v)$ may include conformal factors.

Next consider a photon  moving  in the $Y'$ direction in $S'$ with coordinates $(0,t',t')$. By Axiom~II, this is observed in $S$ with coordinates $(t \sin \theta,t \cos\theta ,t)$ where  $\theta$ is the angle of direction of the photon motion relative to the $Y$ axis. From \eqref{eq:MAB} we obtain
\begin{align*}
t \sin\theta = v A(v)t',\quad 
t=A(v)t',
\quad t\cos \theta=B(v)t'.
\end{align*}
These imply that $ v=\sin{\theta} $ and $B(v)=A(v)\cos{\theta}$. Therefore $|v|=|\sin{\theta}| <1$ so that the IT must be Lorentz-like with $A(v)=\phi(v)\gamma(v)$ from \eqref{eq:phidef} for  conformal factor $\phi(v)$. Furthermore, $\cos\theta=\left(1-v^{2}\right)^{\half}=\gamma(v)^{-1}$ which implies that  $B(v)=\phi(v)$. Thus we have shown that Axiom~II implies that  the IT must be Lorentz-like with $|v|<1$  where
\begin{align}\label{eq:M3d}
	\begin{pmatrix}
		x \\
		t
	\end{pmatrix}
	=
	\phi(v)\gamma(v)\begin{pmatrix}
		1 & v \\
		v  & 1
	\end{pmatrix}
	\begin{pmatrix}
		x'\\
		t'
	\end{pmatrix},\quad y=\phi(v)y'.
\end{align} 
We may repeat the earlier time dilation argument of section~2 to 
obtain \eqref{eq:Ttau} again and conclude that space is anisotropic\footnote{Alternatively, we can repeat Einstein's argument \cite{Einstein.1905}, that if $\phi(v)\neq \phi(-v)$ then the relation $y=\phi(v)y'$ means that space is anisotropic.}  for $\phi(v)\neq 1$.
Thus Axiom~III implies that $\phi(v)=1$ so that  we find \eqref{eq:M3d} is the standard Lorentz transformation for $d=3$. The above arguments easily generalise to all $d>2$ by considering light which in one frame is moving perpendicularly to the direction of relative motion to the other frame.
Thus tachyons cannot exist in a spacetime with $d>2$ where Axioms~I and II both hold.


	\bibliographystyle{amsplain}
	\bibliography{ArXiv_Tachyon}

\end{document}